\documentclass[11pt]{article}

\usepackage[a4paper,margin=1in]{geometry}
\usepackage[T1]{fontenc}
\usepackage[utf8]{inputenc}
\usepackage{lmodern}
\usepackage{microtype}
\usepackage{amsmath,amssymb,amsthm,mathtools}
\usepackage{booktabs}
\usepackage{longtable}
\usepackage{enumitem}
\usepackage{algorithm}
\usepackage{algpseudocodex}
\usepackage{tikz}
\usetikzlibrary{matrix,fit,positioning,backgrounds,arrows.meta}
\usepackage[hidelinks]{hyperref}

\theoremstyle{definition}
\newtheorem{theorem}{Theorem}
\newtheorem{example}{Example}
\newtheorem{lemma}{Lemma}
\newtheorem{corollary}{Corollary}

\newtheorem{remark}{Remark}

\usepackage[capitalise]{cleveref}

\DeclareMathOperator{\rle}{\mathsf{rle}}
\DeclareMathOperator{\ins}{\mathsf{ins}}
\DeclareMathOperator{\del}{\mathsf{del}}
\DeclareMathOperator{\rep}{\mathsf{rep}}
\DeclareMathOperator{\equ}{\mathsf{equ}}
\DeclareMathOperator*{\argmax}{arg\,max}

\newcommand{\N}{N}
\newcommand{\M}{M}
\newcommand{\eps}{\varepsilon}
\newcommand{\Z}{\mathbb{Z}}
\newcommand{\boxij}{\square_{i,j}}

\def\problembox#1{%
	\vspace{2mm}%
	\noindent\fbox{%
	\begin{minipage}{.985\textwidth}%
		#1
	\end{minipage}%
	}%
	\vspace{2mm}%
}

\title{MacCorles: Minimum Alignment Cost Computation on Run-Length Encoded Strings}
\author{Wing-Kai Hon \and Dominik K\"{o}ppl \and Jun-Hong Wang}
\date{}

\begin{document}
\maketitle

\begin{abstract}
We study a tie-breaking variant of the longest common subsequence problem on run-length encoded strings. 
Given two strings, the goal is first to maximize the number of equal aligned character pairs, as in the classical longest common subsequence problem, and then, among all such alignments, to minimize the alignment length. Equivalently, after maximizing the number of equal pairs, we minimize the number of insertions and deletions. We show that this problem admits a simple block-boundary dynamic program. If the input strings have lengths \(\N\) and \(\M\), and their run-length encodings have \(n\) and \(m\) runs, respectively, the algorithm runs in \(O(m\N+n\M)\) time using $O(nm)$ space. 
The algorithm treats every pair of runs as a homogeneous block with an explicit transfer function and stores dynamic-programming values only on run boundaries.
\end{abstract}

\section{Introduction}

The longest common subsequence (LCS) problem is a classical measure of similarity between two strings. It asks for an alignment maximizing the number of equal aligned character pairs. In many applications, however, several LCS alignments may exist, and these alignments may differ substantially in their use of gaps. This motivates the following natural tie-breaking variant: first maximize the number of equal pairs, and then, among all maximum-equality alignments, choose one of minimum alignment length.

\noindent
\begin{minipage}{0.6\linewidth}
Consider two DNA strings
\(
	X=\texttt{TTAACC}
	\text{~and~}
	Y=\texttt{CCAAGG}.
\)
The usual LCS dynamic program (DP) builds the table on the right that gives evidence of two LCSs of length two, \(\texttt{AA}\) and \(\texttt{CC}\),
where rows are indexed by prefixes of \(X\) and columns by prefixes of \(Y\). 
The final value is two, but the table alone does not express that the LCS \(\texttt{AA}\) can be realized without gaps whereas \(\texttt{CC}\) cannot.
In fact, the alignment realizing \(\texttt{CC}\) is
\end{minipage}
\hfill
\begin{minipage}{0.28\linewidth}
\[
\begin{array}{c|rrrrrrr}
	 & \eps & \texttt{C} & \texttt{C} & \texttt{A} & \texttt{A} & \texttt{G} & \texttt{G}\\
\hline
\eps      & 0 & 0 & 0 & 0 & 0 & 0 & 0\\
\texttt{T} & 0 & 0 & 0 & 0 & 0 & 0 & 0\\
\texttt{T} & 0 & 0 & 0 & 0 & 0 & 0 & 0\\
\texttt{A} & 0 & 0 & 0 & 1 & 1 & 1 & 1\\
\texttt{A} & 0 & 0 & 0 & 1 & 2 & 2 & 2\\
\texttt{C} & 0 & 1 & 1 & 1 & 2 & 2 & 2\\
\texttt{C} & 0 & 1 & 2 & 2 & 2 & 2 & 2
\end{array}
\]
\end{minipage}
\[
\begin{array}{cccccccccc}
\texttt{T} & \texttt{T} & \texttt{A} & \texttt{A} & \texttt{C} & \texttt{C} & - & - & - & -\\
- & - & - & - & \texttt{C} & \texttt{C} & \texttt{A} & \texttt{A} & \texttt{G} & \texttt{G}
\end{array}
\]
Such a result can be undesirable under the light that the alignment based on the LCS \(\texttt{AA}\) allows the gap-free alignment
\[
\begin{array}{cccccc}
\texttt{T} & \texttt{T} & \texttt{A} & \texttt{A} & \texttt{C} & \texttt{C}\\
\texttt{C} & \texttt{C} & \texttt{A} & \texttt{A} & \texttt{G} & \texttt{G}
\end{array}
\]
where the first two and last two columns are replacements. 
Our goal is to report the second alignment, which has the same number of equal pairs as the first but fewer insertions and deletions.
We also consider this problem for run-length encoded strings. 
Run-length encoding is particularly suitable for inputs with long maximal runs of equal characters.
 If the first input has length \(\N\) and \(n\) runs, and the second input has length \(\M\) and \(m\) runs, then the usual \(O(\N\M)\)-time dynamic program may be unnecessarily expensive.
 A standard idea in algorithms for comparing run-length encoded strings is to partition the dynamic-programming grid into blocks induced by pairs of runs and to compute only block-boundary values; see, 
e.g., work on RLE-LCS and RLE edit distance~\cite{freschi04longest,ahsan14longest,clifford19rle}.

\paragraph{Our Contribution.}
The contribution of this note is to observe that the minimum-cost LCS tie-breaking objective also has a direct block-boundary formulation.
The key point is that, inside a block induced by two runs, every diagonal step is either \emph{always} an equal-pair step or \emph{always} a replacement.
Hence the optimal contribution of a path segment inside a block depends only on its displacement, not on the individual grid cells visited. 
This leads to an \(O(m\N+n\M)\)-time algorithm.

\paragraph{Roadmap.}
We first show that the problem can be solved with standard means for plain LCS computation using \(O(\N\M)\) time (\cref{thm:quadratic}) and \(O(\min(\N,\M))\) space (\cref{lem:hirschberg}).
Subsequently, we study the problem in RLE-compressed space and give an \(O(m\N+n\M)\)-time algorithm (\cref{thm:main}) using \(O(M+N)\) space (\cref{thm:output-alignment}).
Finally, by switching from a cell-based DP table to a block-based DP table representing cells by piecewise-linear functions, we can improve the space-bound further to \(O(nm)\) by \cref{thm:output-alignment-rlespace}.

\section{Preliminaries}
We work in the word-RAM model with word size \(\Omega(\log(\N+\M))\). 
Working space is measured in machine words, excluding the input representation and excluding the output alignment. 

Let \(X,Y\in\Sigma^*\) be two strings over an alphabet \(\Sigma\), and let
\(
	\N:=|X| \text{~and~} \M:=|Y|
\)
denote their lengths.
For integers \(i,j\), we write \([i..j]=\{i,i+1,\ldots,j\}\) if \(i\le j\), and \([i..j]=\emptyset\) otherwise.

An \emph{alignment path} of \(X\) and \(Y\) is a monotone path \(P\) from \((0,0)\) to \((\N,\M)\) in the dynamic-programming grid (cf.\ the table in the introduction). 
Here \emph{monotone} means that every step is one of
\(
(1,0),\ (0,1),\text{~or~}(1,1).
\)
A horizontal step \((1,0)\) is a deletion, a vertical step \((0,1)\) is an insertion, and a diagonal step \((1,1)\) is either an equal-pair step or a replacement depending on whether the corresponding characters of \(X\) and \(Y\) are equal. Let
\(
\equ(P),\ \ins(P),\ \del(P),\text{~and~}\rep(P)
\)
denote the numbers of equal-pair steps, insertions, deletions, and replacements in \(P\), respectively.
We call \(P\) a \emph{minimum-cost LCS alignment} if it lexicographically maximizes
\[
	\bigl(\equ(P),-(\ins(P)+\del(P))\bigr).
\]
To put in words, \(P\) first maximizes the number of equal pairs and then, among all such alignments, minimizes the number of insertions and deletions.
This is equivalent to minimizing the alignment length after the number of equal pairs has been fixed. 
Indeed, since
\(
	\N+\M = 2\equ(P)+2\rep(P)+\ins(P)+\del(P),
\)
the number of columns in the alignment~\(\ell(P)\) is 
\[
	\ell(P)
	=\equ(P)+\rep(P)+\ins(P)+\del(P)
	=\frac{\N+\M+\ins(P)+\del(P)}{2}.
\]
Thus, for fixed \(X\), \(Y\), minimizing \(\ell(P)\) is the same as minimizing \(\ins(P)+\del(P)\).
We are now ready to state our problem formally.

\problembox{%
\textsc{MinCostLCS}\\
\textbf{Input:} Two strings \(X,Y\in\Sigma^*\).\\
\textbf{Output:} A minimum-cost LCS alignment of \(X\) and \(Y\).
}

For the dynamic program it is convenient to replace the lexicographic objective by a single integer score. Set the weight
\(
	w:=\N+\M+1.
\)
For an alignment path \(P\), define the path scoring function
\(
	s(P):=w \cdot \equ(P)-\ins(P)-\del(P).
\)
The crucial link between $s$ and our goal is given by the following claim.

\begin{lemma}\label{lem:scalar}
An alignment path maximizes \(s\) if and only if it is a minimum-cost LCS alignment.
\end{lemma}

\begin{proof}
For every alignment path \(P\), we have \(0\le \ins(P)+\del(P)\le \N+\M\). Hence, if \(\equ(P)>\equ(P')\), then
\(
	s(P)-s(P')
	\ge w-(\N+\M)
	>0.
\)
Thus, any increase in the number of equal pairs dominates any possible change in the gap term. 
Among all alignment paths $P'$ with the same number of equal pairs, 
maximizing \(s\) is exactly the same as minimizing \(\ins(P')+\del(P')\) over these $P'$. 
This is precisely the lexicographic objective.
\end{proof}

\begin{remark}\label{rem:smaller-w}
The choice \(w=\N+\M+1\) is intentionally simple; smaller values are possible: 
If two alignments differ by one equal pair, the maximum possible disadvantage in the number of gaps is at most \(2\min(\N,\M)-2\). 
Hence, any integer
\(
	w \ge 2\min(\N,\M)-1
\)
is sufficient. 
\end{remark}

Since all scores have magnitude \(O((\N+\M)^2)\), each score fits in \(O(1)\) machine words.
By plugging in the definition of \(s\) into any classic LCS computation algorithm that supports our scoring function~$s$, we immediately obtain the following claim.

\begin{theorem}\label{thm:quadratic}
\textsc{MinCostLCS} can be solved in \(O(\N\M)\) time.
\end{theorem}

\begin{proof}
Let \(F(\alpha,\beta)\) be the maximum scalar score of an alignment path from \((0,0)\) to \((\alpha,\beta)\), where \(\alpha\in[0..\N]\) and \(\beta\in[0..\M]\). Initialize \(F(0,0)=0\), \(F(\alpha,0)=-\alpha\) for \(\alpha\in[1..\N]\), and \(F(0,\beta)=-\beta\) for \(\beta\in[1..\M]\). For \(\alpha\in[1..\N]\) and \(\beta\in[1..\M]\), compute
\[
	F(\alpha,\beta)=\max\bigl(
		F(\alpha-1,\beta)-1,
		F(\alpha,\beta-1)-1,
		F(\alpha-1,\beta-1)+\delta_{\alpha,\beta}
	\bigr),
\]
where \(\delta_{\alpha,\beta}=w\) if \(X[\alpha]=Y[\beta]\), and \(\delta_{\alpha,\beta}=0\) otherwise. The three cases correspond to deletion, insertion, and a diagonal step. Thus \(F(\N,\M)\) is the optimum scalar score, and the desired alignment follows by the usual backtracking. By \Cref{lem:scalar}, it is a minimum-cost LCS alignment. The table has \((\N+1)(\M+1)\) entries and each entry is computed in constant time.
\end{proof}

\begin{example}
Consider again
\[
	X=\texttt{TTAACC}
	\text{~and~}
	Y=\texttt{CCAAGG}.
\]
Here \(\N=\M=6\) and \(w=13\). The scalar table computed by the dynamic program of \Cref{thm:quadratic} is
\[
\begin{array}{c|rrrrrrr}
	 & \eps & \texttt{C} & \texttt{C} & \texttt{A} & \texttt{A} & \texttt{G} & \texttt{G}\\
\hline
\eps      &  0 & -1 & -2 & -3 & -4 & -5 & -6\\
\texttt{T} & -1 &  0 & -1 & -2 & -3 & -4 & -5\\
\texttt{T} & -2 & -1 &  0 & -1 & -2 & -3 & -4\\
\texttt{A} & -3 & -2 & -1 & 13 & 12 & 11 & 10\\
\texttt{A} & -4 & -3 & -2 & 12 & 26 & 25 & 24\\
\texttt{C} & -5 &  9 & 10 & 11 & 25 & 26 & 25\\
\texttt{C} & -6 &  8 & 22 & 21 & 24 & 25 & 26
\end{array}
\]
The final value is \(26=13\cdot 2-0\). Hence the optimum has two equal-pair steps and no gaps, realizing the LCS \(\texttt{AA}\). The alternative LCS \(\texttt{CC}\) also has length two, but any alignment realizing it must pay gaps. For example, the alignment shown in the introduction has eight gaps and therefore scalar score \(13\cdot 2-8=18\).
\end{example}

Finally, we can combine \cref{thm:quadratic} with Hirschberg's divide-and-conquer technique~\cite{hirschberg75linear} to recover an optimum alignment in the same complexities.
This works because our score $s$ admits the property that the scalar score of a monotone path $P$ is equal to the sum of the scalar scores of the two subpaths obtained by splitting $P$ at any point. 
Thus, the optimal alignment can be recovered from the combination of the results of recursively split subproblems.

\begin{theorem}\label{lem:hirschberg}
Given two strings \(X\) and \(Y\), one minimum-cost LCS alignment can be
computed in \(O(\N\M)\) time and \(O(\min(\N,\M))\) working space, excluding
the space needed to write the output alignment.
\end{theorem}
\begin{proof}
We show the case in which the recursion splits \(X\); 
the other case is symmetric. 
Let \(h=\lfloor \N/2\rfloor\). 
Compute, using the linear-space dynamic program, the optimum scores
\(
\textsf{fwd}[\beta]
\)
for aligning \(X[1..h]\) with \(Y[1..\beta]\), for all \(\beta\in[0..\M]\).
Similarly, compute backward scores
\(
\textsf{bwd}[\beta]
\)
for aligning \(X[h+1..\N]\) with \(Y[\beta+1..\M]\), 
for all \(\beta\in[0..\M]\). 
The backward scores are obtained by applying the same dynamic program to the reversed substrings.

Since the score
\(
	s(P)=w \cdot \equ(P)-\ins(P)-\del(P)
\)
is additive over concatenation of alignment paths, any alignment path from
\((0,0)\) to \((\N,\M)\) that crosses the row \(h\) at column \(\beta\) has
score equal to the score of its prefix plus the score of its suffix. 
Hence, an optimal crossing column is obtained by choosing
\[
	\beta^*\in
	\argmax_{\beta\in[0..\M]}
	\bigl(\textsf{fwd}[\beta]+\textsf{bwd}[\beta]\bigr).
\]
There exists an optimal alignment crossing row \(h\) at \(\beta^*\), and it is
obtained by concatenating an optimum alignment of
\(X[1..h]\) with \(Y[1..\beta^*]\) and an optimum alignment of
\(X[h+1..\N]\) with \(Y[\beta^*+1..\M]\).

We recurse on these two subproblems. Each level of the recursion spends
linear-space DP time proportional to the total area of the subproblems at that
level, and these areas form a geometric series bounded by \(O(\N\M)\). The
working space is the space for two score rows, plus recursion overhead, namely
\(O(\M)\) when splitting \(X\). Splitting the shorter string gives
\(O(\min(\N,\M))\) working space.
\end{proof}

\section{Run-Length Encoded Input}

We now assume that the input strings are given in run-length encoded form. We write
\[
	\rle(X)=b_1^{x_1}b_2^{x_2}\cdots b_n^{x_n}
	\text{~and~}
	\rle(Y)=c_1^{y_1}c_2^{y_2}\cdots c_m^{y_m},
\]
where \(b_i,c_j\in\Sigma\), \(x_i,y_j\ge 1\), \(b_i\ne b_{i+1}\), and \(c_j\ne c_{j+1}\). Thus
\(
	\N=\sum_{i=1}^n x_i
	\text{~and~}
	\M=\sum_{j=1}^m y_j.
\)
The substring $b_i^{x_i}$ of $X$ is named \emph{Run}~$i$ of $X$ (analogously for $Y$).
Finally, we define a variant of our initial problem that deviates only in the input format.

\problembox{%
\textsc{MinCostLCS-2RLE}\\
\textbf{Input:} Two run-length encoded strings $X$ and $Y$.
\\
\textbf{Output:} A minimum-cost LCS alignment of the decoded strings \(\rle(X)\) and \(\rle(Y)\).
}

To tackle this problem, we follow the footsteps of Freschi and Bogliolo~\cite{freschi04longest}, 
whose idea is to index the run blocks and compute the dynamic programming table only at the block boundaries.
As we will see, applying the same trick for our extended scoring function is nontrivial but feasible because of its homogeneous properties.

\subsection{Run Blocks and Boundary Values}

From this point on we use the usual table orientation of dynamic programming tables: 
a \emph{point} \((\alpha,\beta)\) represents the prefixes \(X[1..\alpha]\) and \(Y[1..\beta]\), the first coordinate increases downward, and the second coordinate increases to the right. 
For the RLE input strings $X$ and $Y$, define the prefix sums of their exponents 
\[
	p_i:=\sum_{h=1}^{i-1}x_h \text{~for every~}i\in[1..n+1]
	\text{~and~}
	q_j:=\sum_{h=1}^{j-1}y_h \text{~for every~}j\in[1..m+1].
\]
Thus \(p_1=q_1=0\), \(p_{n+1}=\N\), and \(q_{m+1}=\M\). Run \(i\) of \(X\) occupies the character positions \([p_i+1..p_{i+1}]\), and run \(j\) of \(Y\) occupies the character positions \([q_j+1..q_{j+1}]\).

The \emph{run block} induced by runs \(i\) and \(j\) is the point rectangle
\(
	\boxij := [p_i..p_{i+1}]\times[q_j..q_{j+1}]
\)
in the DP grid. It has height \(x_i\) and width \(y_j\). Its unit cells are indexed by
\[
	[p_i+1..p_{i+1}]\times[q_j+1..q_{j+1}],
\]
and every such cell compares the same two characters \(b_i\) and \(c_j\). The \emph{input boundary} of \(\boxij\) is its top and left side, and the \emph{output boundary} is its bottom and right side.

\begin{figure}[t]
\centering
\begin{tikzpicture}[
	cell/.style={minimum width=7mm,minimum height=7mm,draw,inner sep=0pt},
	boundary/.style={line width=.8pt},
	blockfill/.style={fill=blue!10},
	blockline/.style={draw=blue!70!black,line width=1.2pt},
	auxfill/.style={fill=orange!12},
	auxline/.style={draw=orange!70!black,line width=1pt},
	arrow/.style={-{Stealth[length=2mm]},thick},
]
\matrix[matrix of nodes,nodes=cell,row sep=-\pgflinewidth,column sep=-\pgflinewidth] (m) {
	{} & {} & {} & {} & {} & {} & {} & {} & {} & {} \\
	{} & {} & {} & {} & {} & {} & {} & {} & {} & {} \\
	{} & {} & {} & {} & {} & {} & {} & {} & {} & {} \\
	{} & {} & {} & {} & {} & {} & {} & {} & {} & {} \\
	{} & {} & {} & {} & {} & {} & {} & {} & {} & {} \\
	{} & {} & {} & {} & {} & {} & {} & {} & {} & {} \\
};
\coordinate (tbll) at ([yshift=1mm]m-2-3.north west);
\coordinate (tbur) at ([yshift=6mm]m-2-5.north east);
\coordinate (lbll) at ([xshift=-6mm]m-4-3.south west);
\coordinate (lbur) at ([xshift=-1mm]m-2-3.north west);
\begin{scope}[on background layer]
	\node[blockfill,fit=(m-2-3)(m-4-5),inner sep=0pt] {};
	\node[auxfill,fit=(tbll)(tbur),inner sep=0pt] {};
	\node[auxfill,fit=(lbll)(lbur),inner sep=0pt] {};
\end{scope}
\node[blockline,fit=(m-2-3)(m-4-5),inner sep=0pt] (blk) {};
\node[auxline,fit=(tbll)(tbur),inner sep=0pt] (barr) {};
\node[auxline,fit=(lbll)(lbur),inner sep=0pt] (larr) {};
\draw[boundary] (m-1-1.north west) rectangle (m-6-10.south east);
\draw[arrow] ([xshift=2mm]m-2-5.north east) -- ([xshift=2mm]m-4-5.south east) node[midway,right, fill=white] {right output};
\node at (blk.center) {\(\boxij\)};
\node at (barr.center) {\(T\)};
\node at (larr.center) {\(L\)};
\node[left=6mm of m-2-7.east,fill=white,inner sep=0pt] {\(p_i\)};
\node[left=3mm of m-4-7.east,fill=white,inner sep=0pt] {\(p_{i+1}\)};
\node[above=3mm of m-6-3.north,fill=white,inner sep=0pt] {\(q_j\)};
\node[above=3mm of m-6-5.north,fill=white,inner sep=0pt] {\(q_{j+1}\)};
\draw[arrow] ([yshift=-5mm]m-4-3.south west) -- ([yshift=-5mm]m-4-5.south east) node[midway,below,fill=white,inner sep=0pt] {bottom output};
\end{tikzpicture}
\caption{A run block \(\boxij\) as a point rectangle in the ordinary uncompressed DP table, together with the temporary arrays \(T\) and \(L\) storing the top-boundary and left-boundary values. The block-boundary algorithm computes the bottom and right boundary values from these input-boundary values, without materializing all internal table values.}
\label{fig:block-boundary}
\end{figure}
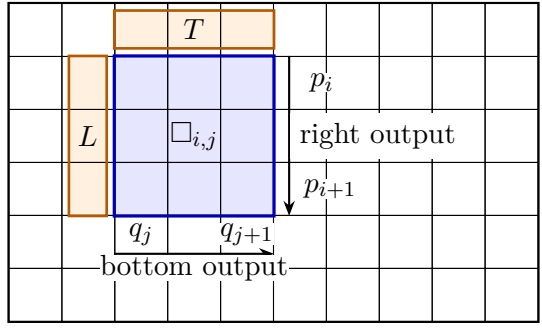

For a local path segment inside \(\boxij\) with displacement \((u,v)\), where \(u\in[0..x_i]\) and \(v\in[0..y_j]\), define the block-transfer score
	\begin{equation}\label{eq:def:tau}
		\tau_{i,j}(u,v)=
		\begin{cases}
			w\cdot\min(u,v)-|u-v| & \text{if~} b_i=c_j,\\[1mm]
			-|u-v| & \text{if~} b_i\ne c_j.
		\end{cases}
	\end{equation}

\begin{lemma}\label{lem:block-transfer}
For any two points \((\alpha,\beta),(\alpha+u,\beta+v)\in\boxij\), the maximum score of a monotone path from \((\alpha,\beta)\) to \((\alpha+u,\beta+v)\) that stays inside \(\boxij\) is \(\tau_{i,j}(u,v)\).
\end{lemma}

\begin{proof}
Any monotone path from \((\alpha,\beta)\) to \((\alpha+u,\beta+v)\) uses at most \(\min(u,v)\) diagonal steps. All remaining displacement must be realized by \(|u-v|\) horizontal or vertical steps. Since horizontal and vertical steps each contribute \(-1\), and diagonal steps contribute either \(w\) if \(b_i=c_j\) or \(0\) if \(b_i\ne c_j\), the best path uses exactly \(\min(u,v)\) diagonal steps. This gives the stated formula.
\end{proof}

Let \(f(\alpha,\beta)\) be the maximum scalar score of an alignment path from \((0,0)\) to \((\alpha,\beta)\). The standard cell-by-cell dynamic program would compute \(f\) at all \((\N+1)(\M+1)\) grid points. We compute \(f\) only on the horizontal and vertical run boundaries, namely all points \((\alpha,q_j)\) and \((p_i,\beta)\).

For a block \(\boxij\), if all values on the input boundary are known, the output values are given by
\begin{align}
	f(p_{i+1},\beta)
	=\max\Bigl(
		&\max_{\beta'\in[q_j..\beta]}
		\bigl(f(p_i,\beta')+\tau_{i,j}(x_i,\beta-\beta')\bigr),
		\notag\\[-1mm]
	&
		\max_{\alpha'\in[p_i..p_{i+1}]}
		\bigl(f(\alpha',q_j)+\tau_{i,j}(p_{i+1}-\alpha',\beta-q_j)\bigr)
	\Bigr)
	\label{eq:bottom}
\end{align}
for all \(\beta\in[q_j..q_{j+1}]\), and
\begin{align}
	f(\alpha,q_{j+1})
	=\max\Bigl(
		&\max_{\beta'\in[q_j..q_{j+1}]}
		\bigl(f(p_i,\beta')+\tau_{i,j}(\alpha-p_i,q_{j+1}-\beta')\bigr),
		\notag\\[-1mm]
	&
		\max_{\alpha'\in[p_i..\alpha]}
		\bigl(f(\alpha',q_j)+\tau_{i,j}(\alpha-\alpha',y_j)\bigr)
	\Bigr)
	\label{eq:right}
\end{align}
for all \(\alpha\in[p_i..p_{i+1}]\). 
These equations try every possible entry point on the input boundary and then use \Cref{lem:block-transfer} for the best internal path from that entry point to the output point.

\subsection{Time Complexity: Processing One Block in Linear Time}\label{sec:time:oneblock}

It remains to show that \Cref{eq:bottom,eq:right} can be evaluated in time linear in the perimeter of the block. 
We state the argument for the bottom side; 
the right side is symmetric. See \cref{fig:block-boundary} for a sketch.

Consider first the contribution from the top boundary to a point \((p_{i+1},q_j+\beta)\), where \(\beta\in[0..y_j]\). 
Create the shorthand \(T[t] := f(p_i,q_j+t)\) for every \(t\in[0..y_j]\), defining the top-boundary values. 
The contribution from the top boundary is
\begin{equation}\label{eq:top-contribution}
	A[\beta]
	:=
	\max_{t\in[0..\beta]}
	\bigl(T[t]+\tau_{i,j}(x_i,\beta-t)\bigr).
\end{equation}
If we evaluate this expression naively for every \(\beta\), then the
computation takes
\(
	\sum_{\beta=0}^{y_j} O(\beta)=O(y_j^2)
\)
time for this contribution alone. We now show how to accelerate the computation of $A$.

\begin{lemma}\label{lem:boundary-acceleration}
	We can evaluate all values \(A[\beta]\), for \(\beta\in[0..y_j]\), in \(O(y_j)\) time.
\end{lemma}
\begin{proof} Let \(d=\beta-x_i\). We have two cases to consider.

\noindent\textbf{Case 1: \(b_i\ne c_j\).~} 
Then
\(
	A[\beta]
	=
	\max_{t\in[0..\beta]}
	\bigl(T[t]-|t-d|\bigr).
\)
Since \(|x_i-\beta+t|=|t-d|\), we have
\[
	T[t]-|t-d|
	=
	\begin{cases}
		T[t]+t-d, & t\in[0..\min(\beta,d)],\\[1mm]
		T[t]-t+d, & t\in[\max(0,d)..\beta].
	\end{cases}
\]
Hence, with the convention that the maximum over an empty interval is
\(-\infty\),
\begin{equation}\label{eq:top-contribution-nonequal}
	A[\beta]
	=
	\max\left(
		\max_{t\in[0..\min(\beta,d)]}(T[t]+t)-d,\;
		\max_{t\in[\max(0,d)..\beta]}(T[t]-t)+d
	\right).
\end{equation}
The first maximum is answered by prefix maxima of \(t \to T[t]+t\).
The second maximum is a sliding-window maximum of \(t \to T[t]-t\),
where the window is \([\max(0,\beta-x_i)..\beta]\). Therefore all values
\(A[\beta]\) can be computed in total \(O(y_j)\) time.

\noindent\textbf{Case 2: \(b_i = c_j\).~} 
Then
\(
	A[\beta]
	=
	\max_{t\in[0..\beta]}
	\bigl(T[t]+w\cdot\min(x_i,\beta-t)-|t-d|\bigr).
\)
Splitting according to whether \(t\le d\), gives
\[
	T[t]+w\cdot\min(x_i,\beta-t)-|t-d|
	=
	\begin{cases}
		T[t]+t+(w+1)x_i-\beta,
			& t\in[0..\min(\beta,d)],\\[1mm]
		T[t]-(w+1)t+(w+1)\beta-x_i,
			& t\in[\max(0,d)..\beta].
	\end{cases}
\]
Therefore
\begin{equation}\label{eq:top-contribution-equal}
	A[\beta]
	=
	\max\left(
		\max_{t\in[0..\min(\beta,d)]}(T[t]+t)+(w+1)x_i-\beta,\;
		\max_{t\in[\max(0,d)..\beta]}(T[t]-(w+1)t)+(w+1)\beta-x_i
	\right).
\end{equation}
The first maximum is again answered by the same prefix maxima.
The second maximum is a sliding-window maximum of
\(T[t]-(w+1)t\). Thus the equal-character case also costs \(O(y_j)\) time
for all bottom-boundary values.
\end{proof}

The contribution from the left side to the bottom side is handled in the same way, using the array \(L[r]=f(p_i+r,q_j)\) for every \(r\in[0..x_i]\), defining the left-boundary values. 
It costs \(O(x_i+y_j)\) time to compute all bottom-side contributions. The right side is symmetric.

\begin{lemma}\label{lem:block-linear}
Given all dynamic-programming values on the top and left side of \(\boxij\), all values on the bottom and right side of \(\boxij\) can be computed in \(O(x_i+y_j)\) time.
\end{lemma}

\subsection{Full Algorithm}

The pseudocode in \cref{alg:main} makes the block-boundary dynamic program explicit. 
The subroutine \textsc{LinearBlockUpdate} in \cref{alg:block} evaluates \Cref{eq:bottom,eq:right} for one block using the prefix, suffix, and sliding maxima described in \Cref{lem:block-linear}.

\begin{algorithm}[t]
\caption{Block-boundary algorithm for \textsc{MinCostLCS-2RLE}}
\label{alg:main}
\begin{algorithmic}[1]
\Require RLE strings \(b_1^{x_1}\cdots b_n^{x_n}\) and \(c_1^{y_1}\cdots c_m^{y_m}\)
\Ensure The optimum scalar score \(f(\N,\M)\)
\State \(\N\gets \sum_{i=1}^n x_i\), \(\M\gets \sum_{j=1}^m y_j\), and \(w\gets \N+\M+1\)
\State \(p_1\gets 0\); \For{\(i\gets 1\) \textbf{to} \(n\)} \State \(p_{i+1}\gets p_i+x_i\) \EndFor
\State \(q_1\gets 0\); \For{\(j\gets 1\) \textbf{to} \(m\)} \State \(q_{j+1}\gets q_j+y_j\) \EndFor
\State Set all boundary values \(f(\alpha,q_j)\) and \(f(p_i,\beta)\) to \(-\infty\)
\For{\(\alpha\gets 0\) \textbf{to} \(\N\)}
	\State \(f(\alpha,0)\gets -\alpha\)
\EndFor
\For{\(\beta\gets 0\) \textbf{to} \(\M\)}
	\State \(f(0,\beta)\gets -\beta\)
\EndFor
\For{\(i\gets 1\) \textbf{to} \(n\)}
	\For{\(j\gets 1\) \textbf{to} \(m\)}
		\State \Call{LinearBlockUpdate}{\(i,j\)}
	\EndFor
\EndFor
\State \Return \(f(\N,\M)\)
\end{algorithmic}
\end{algorithm}

\begin{algorithm}[t]
\caption{Linear update of one block \(\boxij\)}
\label{alg:block}
\begin{algorithmic}[1]
\Procedure{LinearBlockUpdate}{\(i,j\)}
\State Copy the top boundary \(T[t]\gets f(p_i,q_j+t)\) for \(t\in[0..y_j]\)
\State Copy the left boundary \(L[r]\gets f(p_i+r,q_j)\) for \(r\in[0..x_i]\)
\For{\(\beta\gets 0\) \textbf{to} \(y_j\)}
	\State Compute with the computation trick of \cref{lem:boundary-acceleration}
	\begin{equation}\label{eq:def:B}
		B[\beta]\gets
		\max\left(
		\max_{t\in[0..\beta]}\bigl(T[t]+\tau_{i,j}(x_i,\beta-t)\bigr),
		\max_{r\in[0..x_i]}\bigl(L[r]+\tau_{i,j}(x_i-r,\beta)\bigr)
		\right)
	\end{equation}
\EndFor
\For{\(\alpha\gets 0\) \textbf{to} \(x_i\)}
	\State Compute with the computation trick of \cref{lem:boundary-acceleration}
	\begin{equation}\label{eq:def:R}
		R[\alpha]\gets
		\max\left(
		\max_{t\in[0..y_j]}\bigl(T[t]+\tau_{i,j}(\alpha,y_j-t)\bigr),
		\max_{r\in[0..\alpha]}\bigl(L[r]+\tau_{i,j}(\alpha-r,y_j)\bigr)
		\right)
	\end{equation}
\EndFor
\For{\(\beta\gets 0\) \textbf{to} \(y_j\)}
	\State \(f(p_{i+1},q_j+\beta)\gets \max\bigl(f(p_{i+1},q_j+\beta),B[\beta]\bigr)\)
\EndFor
\For{\(\alpha\gets 0\) \textbf{to} \(x_i\)}
	\State \(f(p_i+\alpha,q_{j+1})\gets \max\bigl(f(p_i+\alpha,q_{j+1}),R[\alpha]\bigr)\)
\EndFor
\EndProcedure
\end{algorithmic}
\end{algorithm}

\begin{theorem}\label{thm:main}
\textsc{MinCostLCS-2RLE} can be solved in \(O(m\N+n\M)\) time.
\end{theorem}
\begin{proof}
Correctness follows by induction over the processed blocks. 
The base case is the initialization of the first row and first column:
At a point \((\alpha,0)\),
the second prefix is empty, so the only possible alignment deletes all
\(\alpha\) characters of \(X[1..\alpha]\), giving score \(f(\alpha,0)=-\alpha\).
Similarly, at a point \((0,\beta)\), the only possible alignment inserts all
\(\beta\) characters of \(Y[1..\beta]\), giving score \(f(0,\beta)=-\beta\).
These are the initial boundary values used by the algorithm.

Now assume that all values on the top and left side of a block \(\boxij\) are correct. 
Every path reaching a point on the bottom or right side of \(\boxij\) enters \(\boxij\) for the last time through its top or left side and then stays inside the block until the output point. By \Cref{lem:block-transfer}, the best suffix inside the block depends only on the displacement from this entry point. Therefore \cref{eq:bottom,eq:right} and \cref{alg:block} compute exactly the best value for every output-boundary point. This proves correctness for all block boundaries, in particular for \((\N,\M)\). By \Cref{lem:scalar}, an alignment maximizing this scalar score is a minimum-cost LCS alignment.

By \Cref{lem:block-linear}, block \(\boxij\) is processed in \(O(x_i+y_j)\) time. Summing over all blocks gives
\[
	\sum_{i=1}^n\sum_{j=1}^m O(x_i+y_j)
	=
	O\!\left(m\sum_{i=1}^n x_i+n\sum_{j=1}^m y_j\right)
	=
	O(m\N+n\M).
\]
\end{proof}

\begin{example}
We now revisit the same strings in run-length encoded form:
\[
	\rle(X)=\texttt{T}^2\texttt{A}^2\texttt{C}^2
	\text{~and~}
	\rle(Y)=\texttt{C}^2\texttt{A}^2\texttt{G}^2.
\]
Thus \(\N=\M=6\), \(n=m=3\), and \(w=13\). The block-boundary algorithm does not need to materialize the full scalar table from the example after \Cref{thm:quadratic}. The run-boundary coordinates are
\[
	p_1=q_1=0,
	\quad p_2=q_2=2,
	\quad p_3=q_3=4,
	\quad p_4=q_4=6.
\]
It suffices to compute the following run-boundary entries, which are exactly the corresponding submatrix of the scalar table:
\[
\begin{array}{lc|rrrr}
	 && \epsilon & \texttt{C}^2 & \texttt{A}^2 & \texttt{G}^2\\
	 && 0 & 2 & 4 & 6\\
\hline
\epsilon & 0 &  0 & -2 & -4 & -6\\
\texttt{T}^2 & 2 & -2 &  0 & -2 & -4\\
\texttt{A}^2 & 4 & -4 &  2 & 26 & 24\\
\texttt{C}^2 & 6 & -6 & 22 & 24 & 26
\end{array}
\]
The entry \(22\) at row \(6\), column \(2\) represents the high-scoring partial alignment that has already matched \(\texttt{CC}\) but paid four deletions before doing so. Extending it to the full strings would require four further insertions, giving score \(18\). The final entry \(26\), instead, realizes the gap-free alignment matching \(\texttt{AA}\).
\end{example}

\paragraph{Computational lower bound.}
The stated time complexity in \cref{thm:main} is conditionally optimal up to strongly sublinear factors in the worst case. 
Let us call
\(
	\chi=m\N+n\M
\)
the \emph{compressed cost measure} of the algorithm. 
Our specialized LCS problem generalizes the ordinary LCS problem:
from any minimum-cost LCS alignment, the ordinary LCS length is obtained as the number of equal-pair diagonal steps.
Suppose that, for some constant \(\varepsilon>0\), 
there were an algorithm that solves \textsc{MinCostLCS-2RLE} on all RLE inputs in \(O(\chi^{1-\varepsilon})\) time. 
Given two ordinary strings of decoded lengths \(\N\) and \(\M\), 
we regard them as RLE strings. 
Their numbers of runs satisfy \(n\le \N\) and \(m\le \M\), and therefore
\(
	\chi=m\N+n\M\le 2\N\M.
\)
The assumed algorithm would then compute the ordinary LCS length in \(O((\N\M)^{1-\varepsilon})\) time. 
In particular, for balanced inputs \(\N=\M=L\), 
this gives an \(O(L^{2-2\varepsilon})\)-time algorithm for LCS,
contradicting the known SETH-based conditional lower bound for LCS~\cite{abboud15tight}. 
Hence,
unless SETH fails, 
the present problem admits no \(O(\chi^{1-\varepsilon})\) time algorithm for any constant \(\varepsilon>0\).

\paragraph{Recovering an alignment.}
The theorem gives the value of an optimum alignment. To output an alignment, store for every boundary value the entry point and the case attaining the maximum in \Cref{eq:bottom,eq:right}. Backtracking from \((\N,\M)\) then gives a sequence of block-internal subproblems. Inside a block, the corresponding local path is recovered greedily: use \(\min(u,v)\) diagonal steps and then the remaining \(|u-v|\) gap steps. In equal-character blocks, the diagonal steps are equal-pair steps; in unequal-character blocks, they are replacements. Storing these backpointers uses \(O(m\N+n\M)\) space. 

If only the optimum score is required, the working space can be reduced by processing one strip of blocks at a time:

\begin{lemma}\label{lem:space}
If only the optimum score is required, the block-boundary algorithm can be
implemented in
\(
	O\!\left(
		\min\left(
			\M+\max_{i\in[1..n]} x_i,\;
			\N+\max_{j\in[1..m]} y_j
		\right)
	\right)
	\subset 
	O(\N+\M)
\)
working space.
\end{lemma}
\begin{proof}
We describe the implementation by horizontal strips; the vertical-strip
implementation is symmetric. 
Fix Run \(i\) of \(X\), and process the blocks
\(\boxij\) for \(j=1,2,\ldots,m\) from left to right.

During this scan, we store the top boundary of the current strip, namely the
values \(f(p_i,\beta)\) for all \(\beta\in[0..\M]\), and the bottom boundary
being computed, namely the values \(f(p_{i+1},\beta)\) for all
\(\beta\in[0..\M]\). These two arrays use \(O(\M)\) space.

For the current block \(\boxij\), we additionally store the left boundary
\(
	L[r]=f(p_i+r,q_j)
	\text{~for~} r\in[0..x_i],
\)
and the right boundary computed for the same block as described in \cref{alg:block},
\(
	R[r]=f(p_i+r,q_{j+1})
	\text{~for~} r\in[0..x_i].
\)
These arrays use \(O(x_i)\) space. After finishing block \(\boxij\), the array
\(R\) becomes the left-boundary array for the next block \(\square_{i,j+1}\).

The remaining temporary arrays used by the linear block update have size
\(O(x_i+y_j)\). Since \(y_j\le \M\), the total working space while processing
strip \(i\) is
\(
	O(\M+x_i).
\)
Taking the maximum over all strips gives
\(
	O\!\left(\M+\max_{i\in[1..n]}x_i\right)
\)
working space.

Alternatively, we can process vertical strips: 
By the symmetric argument, this
uses
\(
	O(\N+\max_{j\in[1..m]}y_j)
\)
working space. Choosing the better of the two implementations gives the stated claim.
\end{proof}

With some additional work, we can also apply the space improvements addressed in \cref{lem:hirschberg} to the RLE-setting.
For that, we first state the base case.

\begin{lemma}\label{lem:single-run-base}
Let \(X,Y\in\Sigma^*\), with \(|X|=N\) and \(|Y|=M\). If one of the two
strings consists of a single run, then a minimum-cost LCS alignment of \(X\)
and \(Y\) can be output in \(O(N+M)\) time and \(O(1)\) working space,
excluding the output.
\end{lemma}
\begin{proof}
It suffices to consider the case \(Y=c^M\); the other case is symmetric. 
Let
\(
	h=\#_c(X)
	\)
be the number of occurrences of $c$ in $X$ and let
$\ell=\min(h,M)$, which we claim is the length of an LCS between $X$ and $Y$:
Every equal-pair step must pair an occurrence of \(c\) in \(X\) with a character of \(Y\). 
Hence, no alignment can have more than \(\ell\) equal-pair steps. 
We get this number by pairing \(\ell\) occurrences of \(c\) in \(X\) with \(\ell\) characters of \(Y\). 

It remains to minimize the number of insertions and deletions among alignments with \(\ell\) equal-pair steps. 
Since \(Y\) contains only \(c\)'s, any diagonal step that is not an equal-pair step must pair a non-\(c\) character of \(X\) with a \(c\) from \(Y\), and is therefore a replacement. 
By our scoring function, such a replacement is always better than using one deletion and one insertion. 
Hence, after reserving \(\ell\) characters of \(Y\) for equal-pair steps, 
we should use as many of the remaining characters of \(Y\) as possible for replacement steps. 
The maximum possible number of replacements is
\(
	r=\min(N-h,M-\ell).
\)

Consequently, an optimal alignment has $\ell$ equal-pair steps and $r$ replacement steps.
We can compute such an alignment greedily by counting how many equal-pair steps and how many replacement steps we have already formed by scanning \(X\) from left to right. 
\begin{itemize}
	\item Whenever we see a \(c\) and still need equal-pair steps, we output an equal-pair step.
	\item Whenever we see a non-\(c\) and still need replacement steps, we output a replacement with the next unused character of \(Y\). 
	\item All other characters of \(X\) are deleted. 
\end{itemize}
After the scan, any unused characters of \(Y\) are inserted.

No alignment with \(\ell\) equal-pair steps can have more than \(r\) replacement steps, 
so the produced alignment minimizes the number of insertions and deletions among all LCS alignments. 
\end{proof}

\begin{figure}[t]
	\centering
	\includegraphics{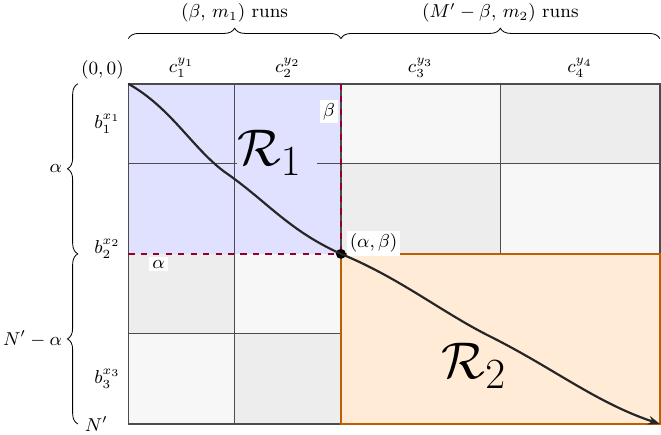}
	\caption{Recursive split for the case \(m'N'\ge n'M'\) in the proof of \cref{thm:output-alignment}.
The second string \(Y'\) is split at a run boundary after decoded length \(\beta\). 
The score vectors at \(\beta\) determine a crossing point \((\alpha,\beta)\), 
yielding the two recursive subproblems 
\(\mathcal R_1 = \bigl(X'[1..\alpha],\,Y'[1..\beta]\bigr)\) 
and 
\(\mathcal R_2 = \bigl(X'[\alpha+1..N'],\,Y'[\beta+1..M']\bigr)\).
	The same original run \(b_2^{x_2}\) appears as the last run of \(\mathcal R_1\) 
	and as the first run of \(\mathcal R_2\).
	In the cost measure its total contribution is \(\beta+(M'-\beta)=M'\).
}
\label{fig:rle-hirschberg-split}
\end{figure}

The main difficulty for a Hirschberg-style divide-and-conquer reconstruction in the aimed RLE space bounds is that we need to 
recurse in the dominating dimension, which can change depending on the subproblem.
This makes it also necessary to cut one of the input strings at a run boundary.
In what follows, we show that, even with these two complications, we can still achieve the desired time and space bounds.

\begin{theorem}\label{thm:output-alignment}
A minimum-cost LCS alignment of two nonempty RLE strings can be output in \(O(m\N+n\M)\) time and \(O(\N+\M)\) working space, excluding the output.
In particular, the LCS string induced by such a minimum-gap-cost LCS alignment can be output within the same bounds.
\end{theorem}

\begin{proof}
We follow \cref{lem:hirschberg} and implement the divide-and-conquer reconstruction by dividing the input strings at run boundaries.

Consider a recursive subproblem with strings \(X'\) and \(Y'\). 
Let \(N' = |X'|\) and \(M' = |Y'|\) be their decoded lengths, and let \(n'\) and \(m'\) be their numbers of runs. 
We define the compressed cost measure of this subproblem as
\(
	\chi=m'N'+n'M'.
\)
If $X$ or $Y$ are empty, the optimum alignment is the unique all-gap alignment.
If \(n'=1\) or \(m'=1\), we solve the subproblem directly by \cref{lem:single-run-base}. 
Thus, in the following, assume \(n',m'\ge 2\).

Without loss of generality suppose that
\(
	m'N'\ge n'M'
\)
 --- otherwise we swap $X'$ and $Y'$ and change them back after the recursive call.
We split the \(Y'\) at a run boundary so that the two parts
contain \(m_1\) and \(m_2\) runs, respectively, with
\(
	m_1, m_2\le \lceil m'/2\rceil\le \frac34 m'.
\)
Let \(\beta\) be the decoded length of the left part, which is also the vertical separator of the decoded part with the decoded length of the right part being \(M'-\beta\).

Following \cref{lem:hirschberg}, we compute the forward scores \(\textsf{fwd}[\gamma]\) for all separator points \((\gamma,\beta)\), where \(\gamma\in[0..N']\). 
Here \(\textsf{fwd}[\gamma]\) is the optimum score for aligning \(X'[1..\gamma]\) with \(Y'[1..\beta]\). 
Similarly, we compute the backward scores \(\textsf{bwd}[\gamma]\), where \(\textsf{bwd}[\gamma]\) is the optimum score for aligning \(X'[\gamma+1..N']\) with \(Y'[\beta+1..M']\). 
We obtain the forward scores by the score-only block-boundary algorithm of \cref{lem:space} and the backward scores by the same algorithm applied to the reversed suffixes. 
Finally, we choose
\[
	\alpha\in
	\argmax_{\gamma\in[0..N']}
		\bigl(\textsf{fwd}[\gamma]+\textsf{bwd}[\gamma]\bigr).
\]
See \cref{fig:rle-hirschberg-split} for an illustration of the split.
The forward computation costs
\(
	O(m_1 N' + n' \beta),
\)
and the backward computation costs
\(
	O(m_2 N' + n'(M'-\beta)).
\)
Since \(m_1+m_2=m'\), their sum is \(O(\chi)\).

By the proof of \cref{lem:hirschberg}, we have that some optimum path crosses the separator at \((\alpha,\beta)\), 
and we can recurse on the two subproblems
\[
	\text{(I):~} (X'[1..\alpha],\,Y'[1..\beta])
	\text{~and~}
	\text{(II):~} (X'[\alpha+1..N'],\,Y'[\beta+1..M']).
\]
\textbf{Time Complexity. }
It remains to bound the total compressed cost of the two recursive subproblems (I) and (II). 
Let \(n_1\) and \(n_2\) denote the numbers of runs of the two corresponding pieces of \(X'\), where the run containing \(\alpha\), if any, is counted once in each recursive subproblem. 
The two recursive subproblems (I) and (II) have the compressed cost measures
\(
m_1\alpha+n_1 \beta
\)
and
\(
m_2(N'-\alpha)+n_2(M'-\beta),
\)
respectively.
Their sum is 
\[
	\bigl(
		m_1\alpha
		+
		m_2(N'-\alpha)
\bigr)
	+
	\bigl(
	n_1 \beta
	+n_2(M'-\beta)
\bigr).
\]
We bound the two groups of summands separately. First,
\[
	m_1\alpha+m_2(N'-\alpha)
	\le \max(m_1,m_2)N'
	\le \frac34 m'N'.
\]
Second, we claim that
\(
	n_1\beta +n_2(M'-\beta)\le n'M'.
\)
To see this, account for the contribution of each run of \(X'\) separately.
Every run of \(X'\) lies completely in one of the two recursive subproblems,
except possibly the run containing the crossing position \(\alpha\). 
A run lying completely on the left contributes \(\beta\), and a run lying completely on the right contributes \(M'-\beta\). 
If \(\alpha\) lies inside a run of \(X'\), then this run appears as a suffix run in the left recursive subproblem and as a prefix run in the right recursive subproblem. 
Its total contribution is then
\(
	\beta+(M'-\beta)=M',
\)
which is exactly the contribution of one run to the parent term \(n'M'\).
Summing over all runs of \(X'\) gives the claimed inequality.

Hence the total compressed cost of the two recursive subproblems is at most
\(
	\frac34 m'N'+n'M'.
\)
Since \(m'N'\ge n'M'\), this is at most
\(
	\frac78(m'N'+n'M')=\frac78\chi.
\)

Thus, if \(T(\chi)\) denotes the reconstruction time for a recursive subproblem
of compressed cost \(\chi\), then
\[
	T(\chi)\le O(\chi)+T(\chi_1)+T(\chi_2),
	\text{~with~}
	\chi_1+\chi_2\le \frac78\chi.
\]
By the Master Theorem, this recurrence solves to \(T(\chi)=O(\chi)\).
For the original input this gives \(O(m\N+n\M)\) time.

\textbf{Space Complexity. }
At any point of the recursion, the score vectors and the temporary arrays of the
score-only block-boundary algorithm have size \(O(N'+M')\le O(\N+\M)\). 
These arrays are discarded before the recursive calls are made. 
The recursion stack stores only recursive subproblem information and has size \(O(\N+\M)\). 
Thus, the working space is \(O(\N+\M)\), excluding the output. 
The LCS string is obtained by outputting the characters of the equal-pair diagonal steps of the produced alignment.
\end{proof}

\section{$O(nm)$-space computation}

By representing the boundary values of each run block in piecewise-affine forms, 
we can apply the Hirschberg trick to reduce the space usage of our algorithm to \(O(nm)\) without increasing the time complexity.
We start with some basic arithmetics on piecewise-affine forms, and subsequently study their computational complexities involved for computing the boundary values of a run block.

\subsection{Piecewise-affine forms}

We say that a function $f$ with an integer domain is $K$-\emph{piecewise-affine} if we can partition the domain of $f$ into $K$ disjoint nonempty integer intervals such that $f$ is affine on each interval.
The \emph{piecewise-affine form} of $f$ is a sorted list of the $K$ intervals, ordered by their left endpoints, augmented by offset and slope values.
In detail, for each interval $[a..b]$ in the partition, we store a pair $(c,d)$ such that $f(x)=c+d\cdot x$ for all $x \in[a..b]$.
Each $([a..b], c, d)$ is called a \emph{piece} of $f$. It is called \emph{affine} if $d \neq 0$, otherwise \emph{constant}.
We use $K$-PAF as a shorthand for a piecewise-affine form~$f$ having $K$ intervals.
Hence, $f$ can be stored in $O(K)$ space.
While we could evaluate $f$ in $O(\log K)$ time by binary search on the list of intervals, we refrain from doing so --- 
we only need linear scans of the list of pieces in our algorithm.

We use PAFs for representing the score array on one side of a block-boundary.
When such a score array is split into pieces belonging to RLE run intervals,
we use any fixed disjoint partition of its domain that assigns a boundary point shared by two consecutive runs to one of the two adjacent run intervals. 
This convention can change the number of pieces only by a constant factor.
In what follows, we translate the arithmetic operations studied in \cref{sec:time:oneblock} from block-boundary cells to PAFs.

\begin{lemma}\label{lem:paf:basic-operations}
Let \(f\) be a $K$-PAF.
Then restriction to an integer subinterval, affine changes of the
argument of the form \(z\mapsto z+c\) and \(z\mapsto c-z\), 
and addition of an affine function preserve \(O(K)\) representation size and can be performed in \(O(K)\) time.
\end{lemma}
\begin{proof}
	Restriction to a subinterval can only delete or shorten pieces. Affine changes
	of the argument of the form \(z\mapsto z+c\) and \(z\mapsto c-z\) translate or
	reverse the interval list. Adding an affine function changes only the affine
	expression stored on each interval. Thus these operations preserve \(O(K)\)
	pieces and take \(O(K)\) time.
\end{proof}

For the next arithmetic operation, we need a helper primitive for computing the upper envelope of constant values supported on ordered non-nested intervals.

\begin{lemma}\label{lem:ordered-interval-envelope}
Let \(I_1=[\lambda_1..\rho_1],\ldots,I_K=[\lambda_K..\rho_K]\) be nonempty
non-nested integer intervals with \(\lambda_1\le\lambda_2\le\cdots\le\lambda_K\) and \(\rho_1\le\rho_2\le\cdots\le\rho_K\). 
Further, let \(c_1,\ldots,c_K\in\Z\). 
Define
\(
	f(z)=\max\{c_h : z\in I_h\},
\)
where \(f(z)\) is undefined if no interval contains \(z\). 
We can write \(f\) as a PAF with \(O(K)\) pieces and can construct $f$ in \(O(K)\) time and space when the intervals are given in order $I_1, \ldots, I_K$.
\end{lemma}
\begin{proof}
We sweep the integer line from left to right over all event positions of the form \(\lambda_h\) and \(\rho_h+1\). 
The deque \(D\) stores only intervals that have started, have not expired, and have not been removed as dominated, 
meaning that there is another later interval with at least as high a value and no earlier expiration event. 
We maintain the deque invariant that the intervals in \(D\) appear in their input order, 
and that their values are strictly decreasing from front to back.

At an event position \(z\), 
we first process all expiration events \(z=\rho_h+1\). 
When such an interval \(I_h\) expires, 
it may already have been removed from the back as dominated; in this case no action is needed. 
If \(I_h\) is still present in \(D\), then it is at the front after all earlier expiring intervals have been removed. 
Indeed, the intervals in \(D\) are in input order, and the right endpoints are nondecreasing. 
Hence, every interval before \(I_h\) in \(D\) has right endpoint at most \(\rho_h\), and has also expired by position \(z\). 
Thus expired intervals are removed from the front of \(D\) at their expiration positions.

After all expirations at \(z\) have been processed, we process all start events \(z=\lambda_h\), in input order. 
Before appending a starting interval \(I_h\) to the back of \(D\), we remove from the back every interval \(I_g\) with \(c_g\le c_h\). 
This is safe because \(I_g\) starts no later than \(I_h\), and the right endpoints are nondecreasing, so \(I_h\) expires no earlier than \(I_g\). 
Since \(c_h\ge c_g\), the interval \(I_g\) can never be the unique maximizer at the current or any later sweep position.

After these removals, we append \(I_h\) to the back. 
The deque invariant is preserved. 
Suppose, for a contradiction, that it is not. 
Before inserting \(I_h\), the invariant held by induction. 
Thus, any violation after the insertion must involve the newly appended interval \(I_h\). 
The input order is still preserved because intervals are processed and inserted in input order. 
Moreover, the right endpoints remain nondecreasing because they are nondecreasing in the input. 
Thus, the only possible violation is that the values are no longer strictly decreasing from front to back. 
Let \(I_g\) be the interval immediately preceding \(I_h\) in \(D\) after the insertion. 
Then such a violation implies \(c_g\le c_h\). 
But all intervals at the back with value at most \(c_h\) were removed before \(I_h\) was appended, a contradiction.

Between two consecutive event positions \(z<z'\), no interval starts or expires.
Hence, the active set is unchanged on the integer interval \([z..z'-1]\).
If \(D\) is empty, then \(f\) is undefined on this interval and no output piece is emitted. 
Otherwise, the maximum value is the value stored at the front of \(D\).
Hence, the active set is unchanged, 
and the maximum value is the value stored at the front of \(D\). 
Whenever this value changes, we start a new output interval. 
Each interval is inserted once, removed from the back at most once, 
and removed from the front at most once. 
Thus, the sweep takes \(O(K)\) time, uses \(O(K)\) space, and produces \(O(K)\) pieces.
\end{proof}

\begin{lemma}\label{lem:paf:max-operations}
Let \(f:[a..b]\to\Z\) be a $K$-PAF\@.
The following operations on $f$ result in an $O(K)$-PAF.
\begin{itemize}
	\item prefix maximum \( F_{\le}(z)=\max_{t\in[a..z]} f(t), \)
	\item suffix maximum \(F_{\ge}(z)=\max_{t\in[z..b]} f(t), \)
	\item fixed-length sliding-window maximum \( F_{\ell}(z)=\max_{t\in[z-\ell..z]\cap[a..b]} f(t) \)
\end{itemize}
Each of them has \(O(K)\) pieces. 
We can construct each of them in \(O(K)\) time and \(O(K)\) space.
\end{lemma}
\begin{proof}
For the prefix maximum, scan the pieces of \(f\) from left to right while maintaining the largest value attained so far.
On one affine piece, the prefix maximum is the pointwise maximum of a constant value and the prefix maximum of the affine function restricted to this piece.
The latter is either the affine function itself or a constant endpoint value, and therefore their pointwise maximum has only \(O(1)\) subpieces on the current piece.
Thus, each input piece creates only \(O(1)\) output pieces.
	The suffix maximum is symmetric.

For the sliding-window maximum, fix a window length \(\ell\). 
For a window \([z-\ell..z]\cap[a..b]\), 
consider its intersection with one piece of \(f\). 
Since an affine function on an integer interval attains its maximum at one of the two endpoints of that interval, 
the maximum over the whole clipped window is attained either at an endpoint of the clipped window, 
or at an endpoint of one of the affine pieces of \(f\) that lies inside the clipped window.

The two clipped-window endpoint contributions have \(O(K)\) pieces: 
away from the boundary of \([a..b]\), 
they are affine changes of \(f\), 
and at the boundary they are constant. 
It remains to handle maxima attained at endpoints of affine pieces. 
Let \(p\) be such an endpoint. 
The value \(f(p)\) can contribute to \(F_{\ell}(z)\) exactly for those window positions \(z\) with
\( p\in[z-\ell..z], \)
that is, for
\( z\in[p..p+\ell]\cap[a..b]. \)
Thus, each piece endpoint gives one constant-valued contribution on an interval of \(z\)-values.

When the piece endpoints are processed from left to right in the domain of \(f\), the corresponding intervals
\(
	[p..p+\ell]\cap[a..b]
\)
have nondecreasing left endpoints and nondecreasing right endpoints. 
Therefore, the requirements of \cref{lem:ordered-interval-envelope} apply to these constant-valued contributions, 
and their upper envelope has \(O(K)\) pieces by  \cref{lem:ordered-interval-envelope}. 
Taking the pointwise maximum of this envelope with the two clipped-window endpoint
contributions gives an \(O(K)\)-piece representation of \(F_{\ell}\).
\end{proof}

\begin{lemma}\label{lem:paf:binary-operations}
Let \(f\) and \(g\) be PAFs on a common integer domain.
Their sum and their pointwise maximum can be computed by a sequential scan of
the common refinement of their sorted interval lists, in time linear in the
number of input pieces plus the number of output pieces.
If the domains are not identical, the same statement applies after restricting
both functions to their common domain.
\end{lemma}
\begin{proof}
For the sum, scan the common refinement of the two interval decompositions and
add the two affine expressions on each refined interval.
For the pointwise maximum, scan the same common refinement.
On each refined interval, the maximum of two affine functions has at most one
crossing point, and therefore at most one additional breakpoint on the integer domain.
Hence, both constructions are output-sensitive linear in the stated sense.
\end{proof}

\subsection{Computing block-boundary values}

Consider one run block $\boxij$ of dimensions $x_i\times y_j$.
Let $f_T:[0..y_j]\to\Z$ be a $K_T$-PAF of the top input block-boundary and let
$f_L:[0..x_i]\to\Z$ be a $K_L$-PAF of the left input block-boundary.
Recall the definition of $\tau_{i,j}$ in \cref{eq:def:tau}.
Ranges over $t$ are always interpreted as subsets of $[0..y_j]$, and ranges
over $r$ as subsets of $[0..x_i]$; empty ranges contribute no candidate.

Let $f_B:[0..y_j]\to\Z$ be the bottom-boundary score array of the block, given
by $f_B(\beta)=B[\beta]$ with $B$ defined in \cref{eq:def:B}. Let
$f_R:[0..x_i]\to\Z$ be the right-boundary score array of the block, given by
$f_R(\alpha)=R[\alpha]$ with $R$ defined in \cref{eq:def:R}.

\begin{lemma}\label{lem:compressed-block-transfer}
The score arrays $f_B$ and $f_R$ can be represented as PAFs with
$O(K_T+K_L+1)$ pieces each. Moreover, they can be constructed together in
\begin{equation*}
	O(K_T+K_L+K_B+K_R+1)
\end{equation*}
time, where $K_B$ and $K_R$ are the numbers of pieces in $f_B$ and $f_R$,
respectively.
\end{lemma}
\begin{proof}
We first consider the bottom-boundary array $f_B$. The contribution from the
top boundary is exactly the array $A[\beta]$ from \cref{eq:top-contribution},
with $T$ represented by the PAF $f_T$. Its two cases are given in
\cref{eq:top-contribution-nonequal,eq:top-contribution-equal}. These formulas
use only affine changes of the argument, addition of affine functions, prefix
maxima, sliding-window maxima, and pointwise maxima applied to $f_T$. Hence,
this top-boundary contribution has $O(K_T+1)$ pieces by
\cref{lem:paf:basic-operations,lem:paf:max-operations,lem:paf:binary-operations}.

The contribution from the left boundary in \cref{eq:def:B} is handled by the
same derivation with the left-boundary variable $r$ in place of the top-boundary
variable $t$. For the bottom boundary, this contribution is split according to
whether $r\le x_i-\beta$ or $r\ge x_i-\beta$. Thus, it is expressed by
prefix and suffix maxima of affine shifts of $f_L$, followed by affine changes
of the output coordinate. It uses only the same PAF operations applied to
$f_L$, and therefore has $O(K_L+1)$ pieces. Taking the pointwise maximum of the
top- and left-boundary contributions gives $f_B$. Thus $f_B$ has
$O(K_T+K_L+1)$ pieces.

The proof for the right-boundary array $f_R$ is symmetric, using
\cref{eq:def:R}: the contribution from the top boundary has $O(K_T+1)$ pieces,
the contribution from the left boundary has $O(K_L+1)$ pieces, and their
pointwise maximum gives $f_R$ with $O(K_T+K_L+1)$ pieces.

The construction of $f_B$ and $f_R$ consists of a constant number of sequential
PAF operations and pointwise maxima. By the output-sensitive bounds in
\cref{lem:paf:binary-operations}, the total time is
$O(K_T+K_L+K_B+K_R+1)$.
\end{proof}

We can now iterate the computation of \cref{lem:compressed-block-transfer} horizontally and vertically to obtain global complexities.
We do not apply the following results on the global input strings $X$ and $Y$, but rather on rectangular subproblems emerging during the recursive divide-and-conquer computation in the Hirschberg algorithm.

\begin{lemma}\label{lem:compressed-row-transfer}
We can compute a row-boundary score array represented as a function of the second coordinate by a row-wise sweep using \(O(nm)\) working space. 
During such a computation, every active horizontal score-array segment restricted to one run interval of $Y$ has \(O(n)\) pieces.
\end{lemma}
\begin{proof}
Our row-wise computation processes the run blocks in increasing order of the runs of $X$. 
At any time we store the previous row of horizontal score-array segments, 
the already constructed part of the current row, 
and the current vertical score-array segment used as the left input score array of the next block.

We prove the following invariant for row \(i\): 
While processing the \(i\)-th run of $X$, 
every active horizontal score-array segment restricted to one run interval of $Y$ has \(O(i)\) pieces, 
and every active vertical score-array segment restricted to the current run interval of the first string has \(O(i)\) pieces.

For \(i=1\), 
the top boundary and the left boundary of the dynamic-programming table are affine on each run interval, 
so the invariant holds. 
Suppose that the invariant holds before processing block \((i,j)\). 
The top score array of this block is a horizontal segment from the previous row, 
and therefore has \(O(i)\) pieces on the current run interval of $Y$. 
The left score array is either the left boundary of the whole table, 
or the right score array produced by block \((i,j-1)\); 
by the invariant, it has \(O(i)\) pieces on the current run interval of $X$. 
By \Cref{lem:compressed-block-transfer}, 
the bottom and right output score arrays of block \((i,j)\) have \(O(i)\) pieces.
Thus, the invariant is preserved while scanning the row. 
After row \(i\) has been processed, 
the bottom score arrays become the top score arrays for row \(i+1\),
and they have \(O(i)\subseteq O(i+1)\) pieces. 
The invariant follows by induction.

Consequently, after all \(n\) rows have been processed, the desired horizontal score array has \(O(n)\) pieces on each of the \(m\) run intervals of the second string. 
Its total representation size is \(O(nm)\). 
The active score arrays stored by the row-wise sweep have the same asymptotic size, and hence the working space is \(O(nm)\).
\end{proof}

The vertical statement follows by exchanging the roles of $X$ and $Y$, and using a column-wise sweep.

\begin{corollary}\label{cor:compressed-column-transfer}
We can compute a column-boundary score array represented as a function of the first coordinate by a column-wise sweep using \(O(nm)\) working space. 
During such a computation, every active vertical score-array segment restricted to one run interval of the first string has \(O(m)\) pieces.
\end{corollary}
\begin{proof}
	We process the blocks column-by-column
	while maintaining the invariant that, 
	while processing column \(j\), 
	every active vertical score array segment restricted to one run interval of $X$ has \(O(j)\) pieces. 
	Hence, a vertical output score array has \(O(m)\) pieces on each of the \(n\) run intervals of $X$, 
	and therefore has \(O(nm)\) pieces in total. 
	The working space is \(O(nm)\).
\end{proof}

\begin{lemma}\label{lem:global-score-array-complexity}
We can compute the optimum score in \(O(mN+nM)\) time within \(O(nm)\) working space.
\end{lemma}
\begin{proof}
	By \cref{lem:compressed-row-transfer,cor:compressed-column-transfer}, 
	we can compute the horizontal and vertical output score arrays in \(O(nm)\) working space.
It remains to bound the time. Consider one block \((i,j)\), of dimensions
\(x_i\times y_j\). Let \(K_T,K_L,K_B,K_R\) be the numbers of pieces in its top,
left, bottom, and right score arrays. By
\Cref{lem:compressed-block-transfer}, the transfer through the block takes
\( O(K_T+K_L+K_B+K_R+1) \)
time. 
Since a piecewise-affine function on an integer interval of length \(y_j\) (resp.\ \(x_i\)) has at most \(y_j+1\) pieces (resp.\ \(x_i+1\) pieces), 
this is
\( O(x_i+y_j+1). \)
Summing over all blocks gives
\[
	\sum_{i=1}^{n}\sum_{j=1}^{m} O(x_i+y_j+1) = O(mN+nM+nm) = O(mN+nM),
\]
because \(N\ge n\) and \(M\ge m\). This completes the proof.
\end{proof}

Since \cref{lem:global-score-array-complexity} can be applied on any subproblem,
plugging the compressed score-array computation into \cref{thm:output-alignment} gives the following result.

\begin{theorem}\label{thm:output-alignment-rlespace}
A minimum-cost LCS alignment of two nonempty RLE strings can be output in $O(m\N+n\M)$ time and $O(nm)$ working space, excluding the output.
In particular, the LCS string induced by such a minimum-gap-cost LCS alignment can be output within the same bounds.
\end{theorem}
\begin{proof}
We use the same Hirschberg-style reconstruction as in the proof of \Cref{thm:output-alignment}. 
The only difference is how the score arrays used for choosing the crossing point are represented and computed.

Consider one recursive subproblem with decoded lengths \(N'\) and \(M'\), and
with \(n'\) and \(m'\) runs. 
In the proof of \Cref{thm:output-alignment}, the
algorithm computes one forward and one backward score array on the chosen
separator. We compute the same arrays, but store them in piecewise-affine form.
By \cref{lem:compressed-row-transfer,cor:compressed-column-transfer}, 
applied to the current recursive subproblem, 
the separator arrays are stored in \(O(n'm')\) working space. 
The time bound follows from the block-wise summation in the proof of \cref{lem:global-score-array-complexity}, 
applied to the corresponding prefix or suffix subproblem: 
a horizontal separator array is computed by a row-wise sweep, and a vertical separator array by a column-wise sweep, in
\( O(m'N' + n'M') \)
time.

The forward score array has to be kept while the backward score array is
computed. If the separator splits $Y$ into parts with \(m_1\) and
\(m_2\) runs, then the two arrays use
\( O(n'm_1+n'm_2)=O(n'm') \)
space in total. 
The case where the separator splits $X$ is symmetric. 
The backward array is computed on reversed suffixes. 
Before scanning it together with the forward array, 
we apply the corresponding affine change of coordinates, 
for instance $\gamma \mapsto N'-\gamma$ for a vertical separator. 
This operation preserves the number of PAF pieces. 
The crossing point is then found by scanning the common refinement of the two piecewise-affine representations of the forward and backward arrays; 
this takes time linear in their total size, hence \(O(n'm')\), 
which is dominated by \(O(m'N'+n'M')\).

Correctness is unchanged from \Cref{thm:output-alignment}: 
every monotone alignment path crosses the chosen separator, 
the score is additive under concatenation of paths, 
and the chosen crossing point maximizes the sum of the forward and backward scores. 
Hence, the recursive decomposition outputs a minimum-cost LCS alignment.

The time recurrence is the same as in the proof of \Cref{thm:output-alignment}. 
With
\( \chi = m'N' + n'M'\),
the two recursive subproblems have total compressed cost at most a constant fraction of \(\chi\), 
and the work spent at the current subproblem is \(O(\chi)\). 
Therefore the total time over the recursion is \( O(m\N+n\M). \)

For the space bound, every recursive subproblem uses at most
\( O(n'm')\le O(nm) \)
working space for its active score arrays. These arrays are discarded before
the recursive calls are made. 
The recursion stack stores only subproblem information and is in \(O(\log n+\log m)\) space, hence within \(O(nm)\) space for nonempty RLE inputs. 
The output is streamed in the same left-to-right order as in \Cref{thm:output-alignment}. 
Thus, the working space is \(O(nm)\), excluding the output.
\end{proof}

The output of \cref{thm:output-alignment-rlespace} should be streamed. 
If the whole alignment is stored explicitly, then an additional $\Theta(\N+\M)$ space is necessary in the worst case.

\section{Conclusion and Discussion}

Our algorithm translates known techniques for plain LCS computation on RLE strings~\cite{freschi04longest} to the minimum-cost LCS tie-breaking objective, resulting in comparable time bounds.
This translation works for the following two reasons.

First, the scoring function $s$ is additive over the steps of an alignment path and is homogeneous inside each RLE block.
Thus, for any two points inside a block, the order of the internal steps is irrelevant; only the number of diagonal steps matters.
Second, a diagonal step is always at least as good as replacing it by one downward and one rightward step.
Hence, an optimal local path uses as many diagonal steps as possible.
Consequently, the block-transfer score \(\tau_{i,j}\) is piecewise affine with only constantly many pieces.
This allows us to compute the boundary-to-boundary maximizations in linear time in the block perimeter, instead of evaluating all pairs of boundary points naively.
A corollary is that other scalar functions with the same two properties can be computed in the same time bound with the same techniques.

A question is whether the information computed at the boundaries can themselves be represented in compressed form to improve space usage.
Another research direction is to improve the time complexity.
Another result for plain LCS computation is due to Clifford et al.~\cite{clifford19rle}, 
who obtain an \(O(mn\log(mn))\)-time algorithm using more refined data structures for the information at the boundaries. 
Since the scoring function used here is piecewise linear, 
it seems plausible to us that similar ideas can further improve the \(O(m\N+n\M)\) time bound. 
We leave this as a direction for future work.

\vspace{1em}
\textbf{Author Contributions and AI Disclosure. }
The authors acknowledge the use of OpenAI's ChatGPT-5.5 Plus to assist with clarity, language, and soundness, and the generation of \cref{tab:notation}.
All AI-generated outputs were thoroughly reviewed, edited, and verified for accuracy and intellectual integrity by the authors.

\bibliographystyle{plain}
\bibliography{literature}

\clearpage
\appendix
\section{Summary of Notation}
\label{app:notation}

\begin{longtable}{@{}p{.22\linewidth}p{.72\linewidth}@{}}
\caption{Variable names and symbols used in the paper.}\label{tab:notation}\\
\toprule
\textbf{Symbol} & \textbf{Meaning}\\
\midrule
\endfirsthead
\toprule
\textbf{Symbol} & \textbf{Meaning}\\
\midrule
\endhead
\midrule
\multicolumn{2}{r@{}}{Continued on the next page}\\
\endfoot
\bottomrule
\endlastfoot
\(\Sigma\) & Alphabet.\\
\(X,Y\) & Input strings. In the RLE setting these are the decoded strings represented by the two run-length encodings.\\
\(\N,\M\) & Lengths of the input strings, i.e., \(\N=|X|\) and \(\M=|Y|\).\\
\([i..j]\) & Integer interval \(\{i,i+1,\ldots,j\}\) if \(i\le j\), and the empty set otherwise.\\
\(P\) & Monotone alignment path from \((0,0)\) to \((\N,\M)\).\\
\((1,0)\) & Downward step of an alignment path; represents a deletion.\\
\((0,1)\) & Rightward step of an alignment path; represents an insertion.\\
\((1,1)\) & Diagonal step of an alignment path; represents either an equal-pair step or a replacement.\\
\(\equ(P)\) & Number of equal-pair steps in the alignment path \(P\).\\
\(\ins(P)\) & Number of insertions in the alignment path \(P\).\\
\(\del(P)\) & Number of deletions in the alignment path \(P\).\\
\(\rep(P)\) & Number of replacements in the alignment path \(P\).\\
\(\ell(P)\) & Alignment length, i.e., the number of columns in the alignment represented by \(P\).\\
\(w\) & Scalar weight used to encode the lexicographic objective as an integer score; the paper sets \(w=\N+\M+1\).\\
\(s(P)\) & Scalar score of \(P\), defined as \(s(P)=w\equ(P)-\ins(P)-\del(P)\).\\
\(F(\alpha,\beta)\) & Maximum scalar score of an alignment path from \((0,0)\) to \((\alpha,\beta)\) in the ordinary \(O(\N\M)\)-time dynamic program.\\
\(\alpha,\beta\) & Grid coordinates in the ordinary dynamic program; \(\alpha\in[0..\N]\) and \(\beta\in[0..\M]\).\\
\(\delta_{\alpha,\beta}\) & Diagonal contribution in the ordinary dynamic program: \(w\) if \(X[\alpha]=Y[\beta]\), and \(0\) otherwise.\\
\(\rle(S)\) & Run-length encoding of a string \(S\).\\
\(n,m\) & Numbers of runs in \(\rle(X)\) and \(\rle(Y)\), respectively.\\
\(b_i,c_j\) & Characters of the \(i\)th run of \(X\) and the \(j\)th run of \(Y\), respectively.\\
\(x_i,y_j\) & Lengths of the \(i\)th run of \(X\) and the \(j\)th run of \(Y\), respectively.\\
\(p_i,q_j\) & Prefix coordinates of run boundaries: \(p_i=\sum_{h=1}^{i-1}x_h\) and \(q_j=\sum_{h=1}^{j-1}y_h\). Run \(i\) of \(X\) occupies character positions \([p_i+1..p_{i+1}]\), and run \(j\) of \(Y\) occupies character positions \([q_j+1..q_{j+1}]\).\\
\(\square_{i,j}\) & Run block induced by the \(i\)th run of \(X\) and the \(j\)th run of \(Y\); it is the point rectangle \([p_i..p_{i+1}]\times[q_j..q_{j+1}]\) in the DP grid, whose unit cells are indexed by \([p_i+1..p_{i+1}]\times[q_j+1..q_{j+1}]\).\\
\(T\) & Temporary array storing top-boundary values in \Cref{fig:block-boundary} and \Cref{alg:block}.\\
\(u,v\) & Local displacement inside a run block in the \(X\)- and \(Y\)-directions, respectively.\\
\(\tau_{i,j}(u,v)\) & Block-transfer score for a path segment with displacement \((u,v)\) inside \(\square_{i,j}\).\\
\(f(\alpha,\beta)\) & Maximum scalar score stored on run boundaries by the block-boundary algorithm.\\
\(\alpha',\beta'\) & Entry-point coordinates on the input boundary of a run block in the block-transfer equations.\\
\(r,t\) & Local indices on the left and top boundary of a run block, respectively.\\
\(T[t]\) & Temporary array storing top-boundary values of the current block: \(T[t]=f(p_i,q_j+t)\).\\
\(L[r]\) & Temporary array storing left-boundary values of the current block: \(L[r]=f(p_i+r,q_j)\).\\
\(B[\beta]\) & Temporary array storing the newly computed bottom-boundary values of the current block.\\
\(R[\alpha]\) & Temporary array storing the newly computed right-boundary values of the current block.\\
\(\eps\) & Empty prefix, used as a row or column label in example dynamic-programming tables.\\
\end{longtable}

\end{document}